\documentclass{cccg19}

\usepackage{verbatim}
\usepackage{graphicx,amssymb,amsmath}
\usepackage{paralist}
\usepackage{wrapfig}
\usepackage{subcaption}
\usepackage{enumitem}
\usepackage{xspace}
\usepackage{multicol}
\usepackage[table]{xcolor}
\usepackage[ruled,vlined,linesnumbered]{algorithm2e}

\graphicspath{{img/}}

\newcommand{\setP}{\ensuremath{P}\xspace}
\newcommand{\setR}{\ensuremath{R}\xspace}

\newcommand{\wwlog}{without loss of generality\xspace}

\newcommand{\bigOh}{\mathcal{O}}

\newcommand{\numNE}{\kappa}
\newcommand{\numBorder}{k}

\newcommand{\eps}{\varepsilon}

\newcommand{\metricSet}{\mathbb{R}^d}
\newcommand{\metricFunc}{\textup{\textsf{d}}}

\newcommand{\dist}[2]{\metricFunc(#1,#2)}
\newcommand{\nn}[1]{\textup{nn}(#1)}
\newcommand{\nenemy}[1]{\textup{ne}(#1)}
\newcommand{\dnn}[1]{\metricFunc_\textup{nn}(#1)}
\newcommand{\dne}[1]{\metricFunc_\textup{ne}(#1)}

\newcommand{\NN}{\textup{NN}\xspace}
\newcommand{\NE}{\textup{NE}\xspace}

\newcommand{\CNN}{\textup{CNN}\xspace}
\newcommand{\FCNN}{\textup{FCNN}\xspace}
\newcommand{\NET}{\textup{NET}\xspace}
\newcommand{\MSS}{\textup{MSS}\xspace}
\newcommand{\VSS}{\textup{VSS}\xspace}
\newcommand{\RSS}{\textup{RSS}\xspace}

\definecolor{yellowcd}{RGB}{252, 229, 30}
\definecolor{bluecd}{RGB}{51, 0, 68}

\title{Guarantees on Nearest-Neighbor Condensation heuristics\thanks{Research supported by NSF grant CCF--1618866.}}

\author{Alejandro Flores-Velazco\thanks{University of Maryland, College Park, {\tt afloresv@cs.umd.edu}}
        \and
        David Mount\thanks{University of Maryland, College Park, {\tt mount@umd.edu}}}

\index{Gumble, Barney}
\index{Szyslak, Moe}

\begin{document}
\thispagestyle{empty}
\maketitle


\noindent\textbf{Abstract}~ 
The problem of nearest-neighbor (\NN) condensation aims to reduce the size of a training set of a nearest-neighbor classifier while maintaining its classification accuracy. Although many condensation techniques have been proposed, few bounds have been proved on the amount of reduction achieved. In this paper, we present one of the first theoretical results for practical \NN condensation algorithms. We propose two condensation algorithms, called \RSS and \VSS, along with provable upper-bounds on the size of their selected subsets. Additionally, we shed light on the selection size of two other state-of-the-art algorithms, called \MSS and \FCNN, and compare them to the new algorithms.

\section{Introduction}

In machine learning, \emph{classification} involves a training set $\setP \subset \metricSet$ of $n$ \emph{labeled} points in Euclidean space.
The label $l(p)$ of each point $p \in \setP$ indicates the \emph{class} to which the point belongs to, partitioning of \setP into a finite set of \emph{classes}. Given an \emph{unlabeled} query point $q \in \metricSet$ the goal of a \emph{classifier} is to predict $q$'s label using \setP.

The \emph{nearest-neighbor} (\NN) \emph{rule} is among the best-known classification techniques~\cite{fix_51_discriminatory}. It classifies a query point $q$ with the label of its closest point in $\setP$, according to some metric. Throughout, we will assume the Euclidean $\ell_2$ metric.
Despite its simplicity, the \NN rule exhibits good classification accuracy both experimentally and theoretically ~\cite{stone1977consistent,Cover:2006:NNP:2263261.2267456,devroye1981inequality}.
However, it's often criticized due to its high space and time complexities. This raises the question of whether
it is possible to replace \setP with a significantly smaller subset without affecting the classification accuracy under the \NN rule. This problem is called \emph{nearest-neighbor~condensation}. In this paper we propose two new \NN condensation algorithms and analyze their worst-case performance.

%
%
%
\subsection{Preliminaries}

For any point $p \in \setP$, define an \emph{enemy} of $p$ to be any point in \setP of different class than $p$. The \emph{nearest enemy} (\NE) of $p$, denoted $\nenemy{p}$, is the closest such point, and its distance from $p$, called the \emph{\NE distance}, is denoted as $\dne{p} = \dist{p}{\nenemy{p}}$. Similarly, denote the \NN distance as $\dnn{p} = \dist{p}{\nn{p}}$. Define the \emph{\NE ball} of $p$ to be the ball centered at $p$ with radius $\dne{p}$. Let $\numNE$ denote the number of distinct \NE points of \setP.

A point $p \in \setP$ is called a \emph{border point} if it is incident to an edge of the Delaunay triangulation of $\setP$ whose opposite endpoint is an enemy of $p$. Otherwise, $p$ is called an \emph{internal point}. By definition, the border points of \setP completely characterize the portion of the Voronoi diagram that separates Voronoi cells of different classes. Let $\numBorder$ denote the number of border points of \setP.

%
%
%
\subsection{Related work}

\begin{figure*}[t]
    \centering
    \begin{subfigure}[b]{.20\linewidth}
        \centering\includegraphics[width=.95\textwidth]{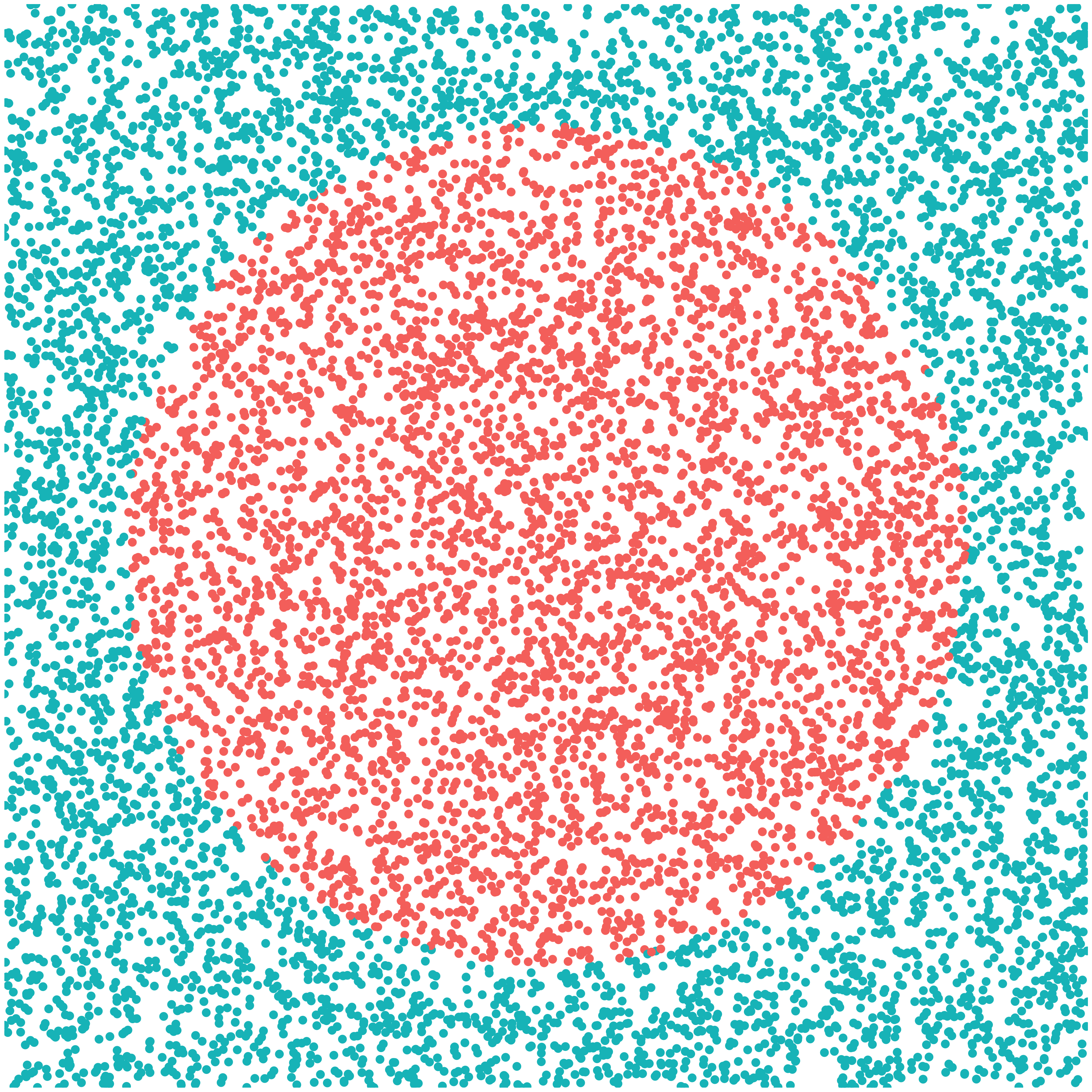}
        \caption{Set \setP ($10^4$\,pts)}\label{fig:circle:dataset}
    \end{subfigure}%
    \begin{subfigure}[b]{.20\linewidth}
        \centering\includegraphics[width=.95\textwidth]{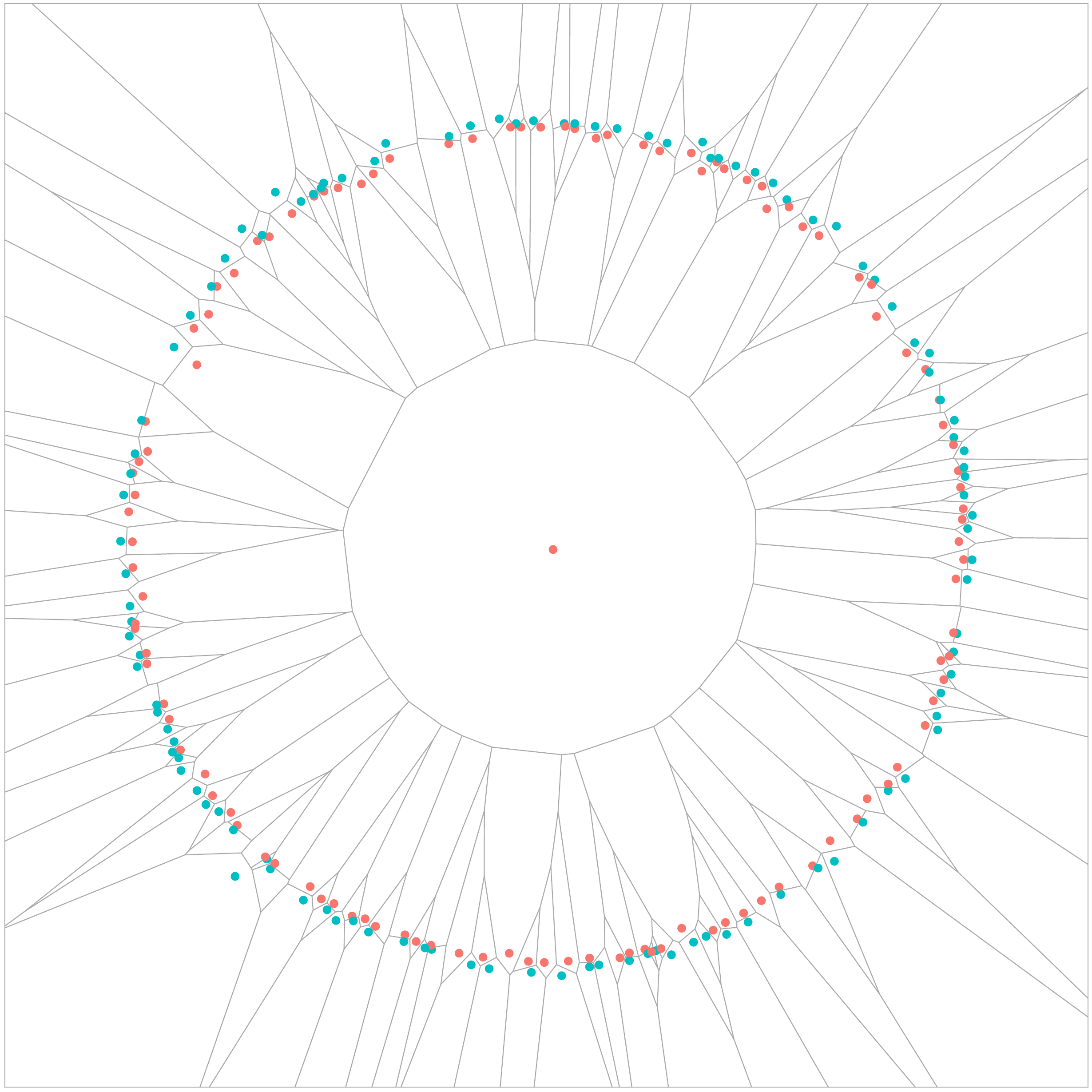}
        \caption{\FCNN (222 pts)}\label{fig:circle:fcnn}
    \end{subfigure}%
    \begin{subfigure}[b]{.20\linewidth}
        \centering\includegraphics[width=.95\textwidth]{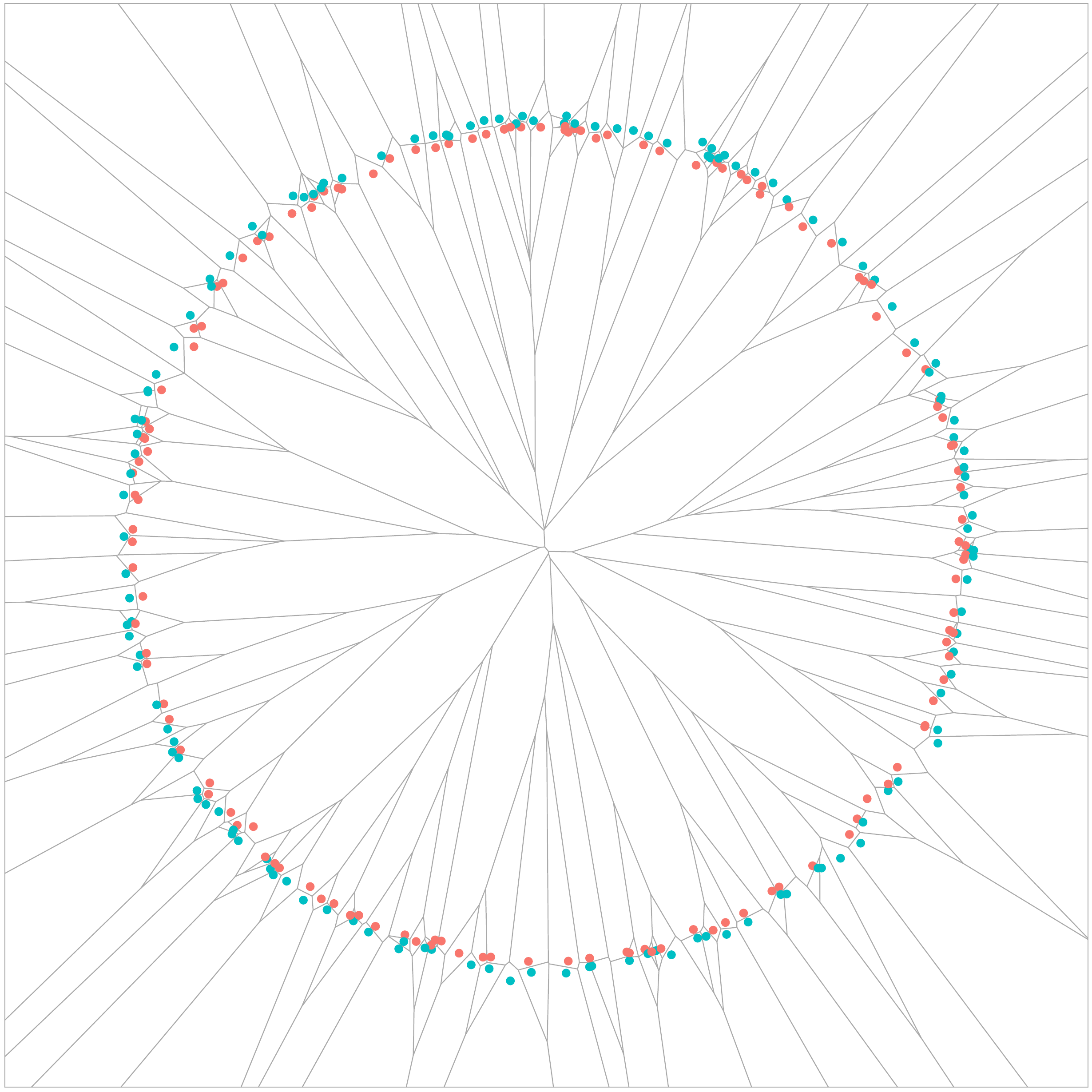}
        \caption{\MSS (272 pts)}\label{fig:circle:mss}
    \end{subfigure}%
    \begin{subfigure}[b]{.20\linewidth}
        \centering\includegraphics[width=.95\textwidth]{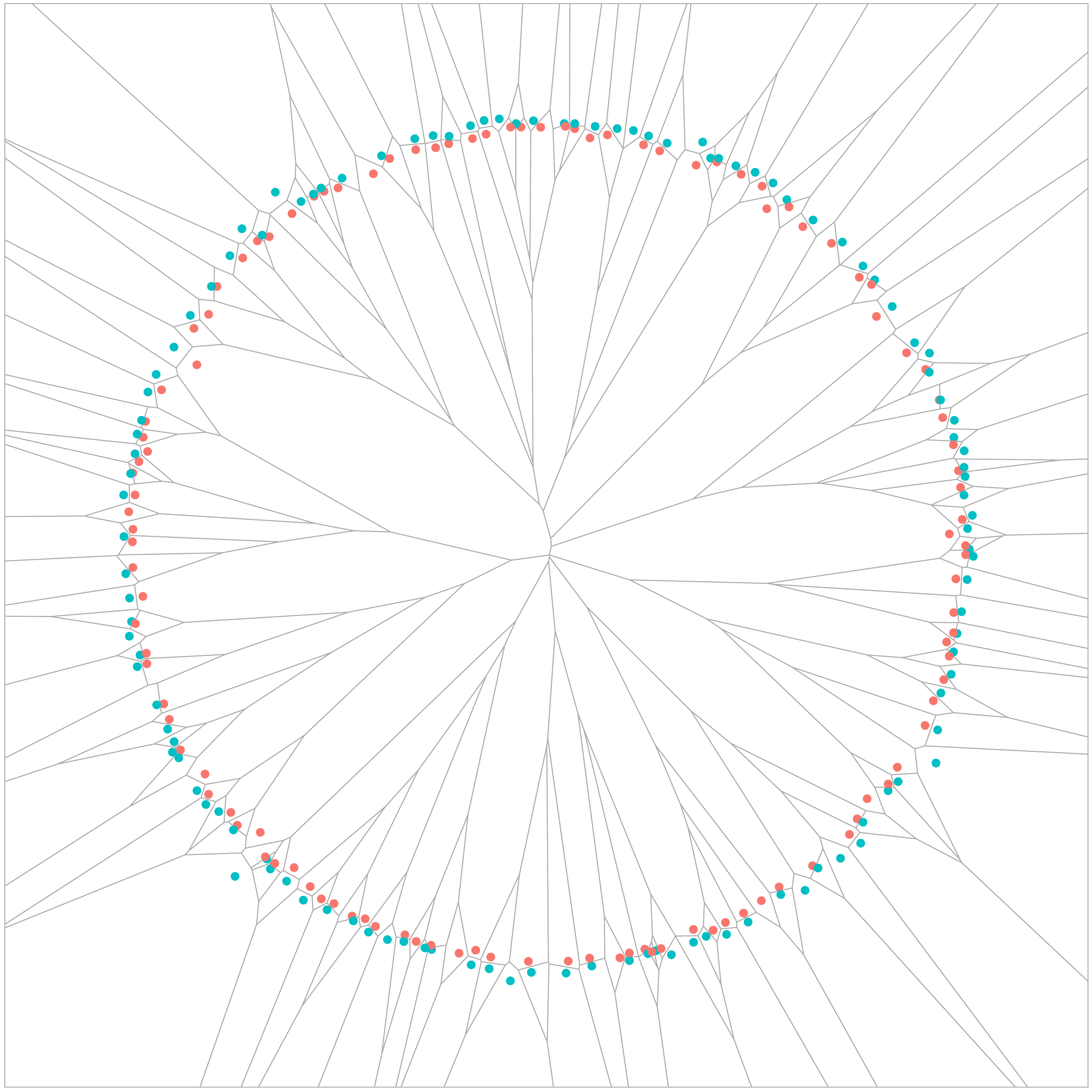}
        \caption{\RSS (233 pts)}\label{fig:circle:rss}
    \end{subfigure}%
    \begin{subfigure}[b]{.20\linewidth}
        \centering\includegraphics[width=.95\textwidth]{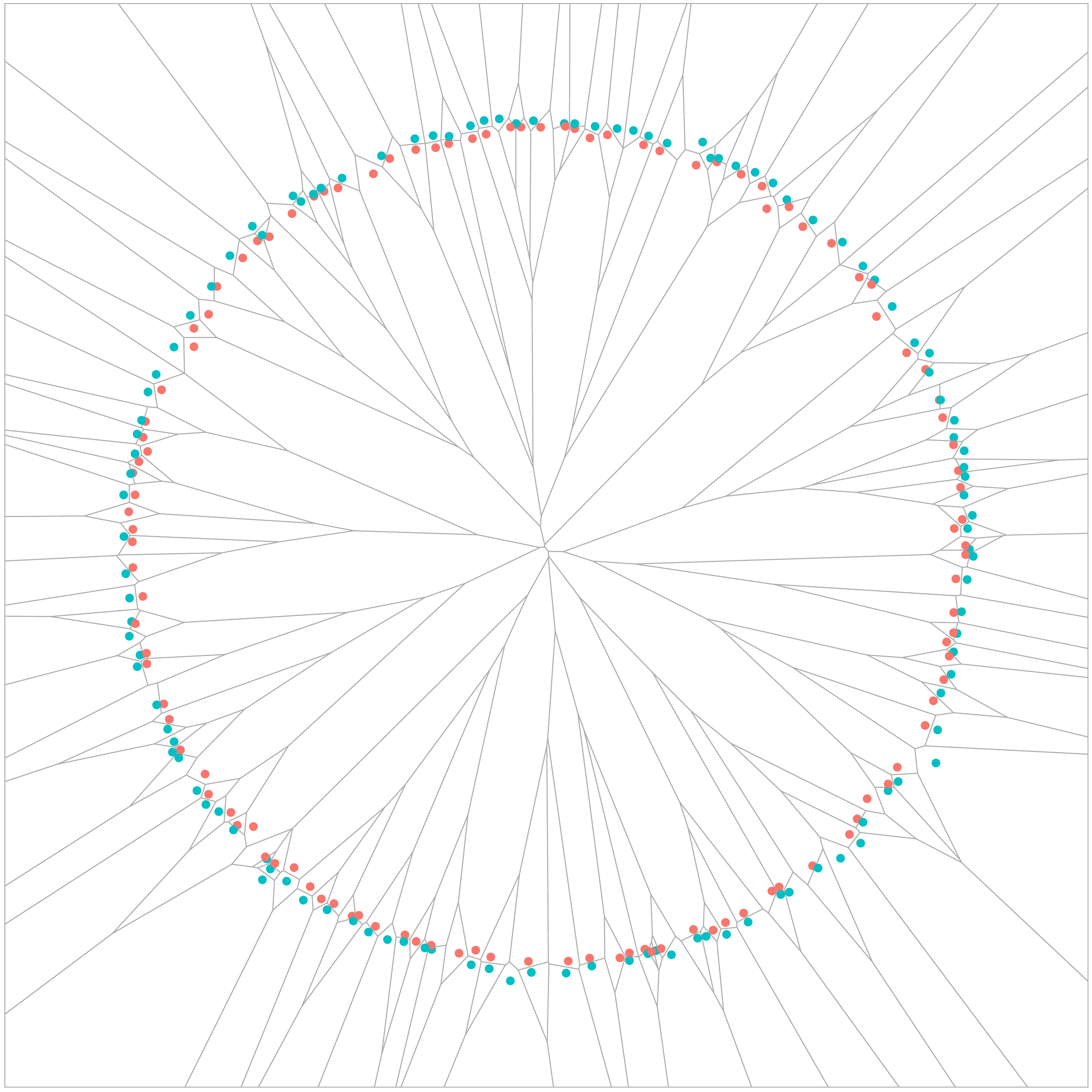}
        \caption{\VSS (233 pts)}\label{fig:circle:vss}
    \end{subfigure}%
    
    \bigskip
    \begin{subfigure}[b]{.20\linewidth}
        \centering\includegraphics[width=.95\textwidth]{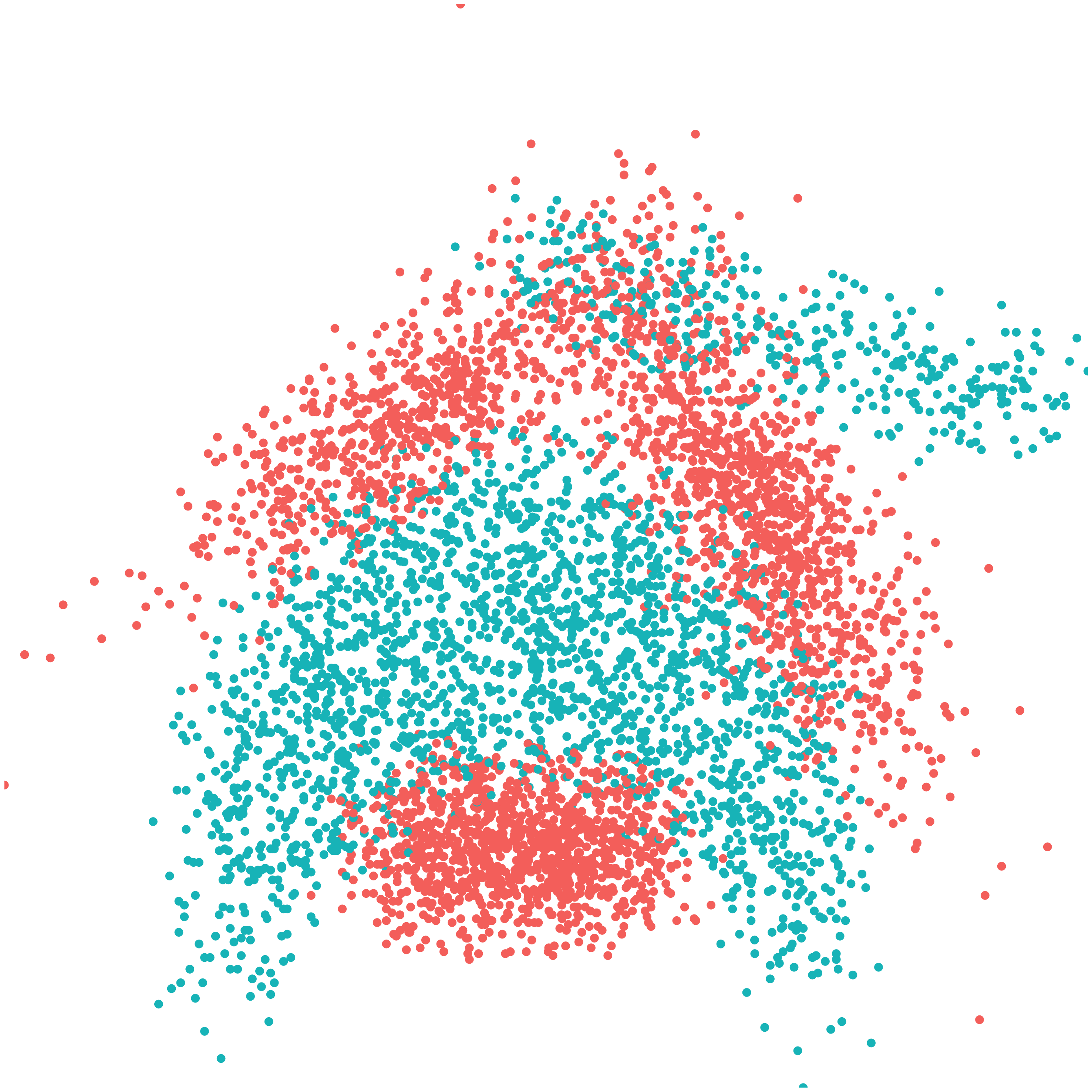}
        \caption{Set \setP (5300 pts)}\label{fig:banana:dataset}
    \end{subfigure}%
    \begin{subfigure}[b]{.20\linewidth}
        \centering\includegraphics[width=.95\textwidth]{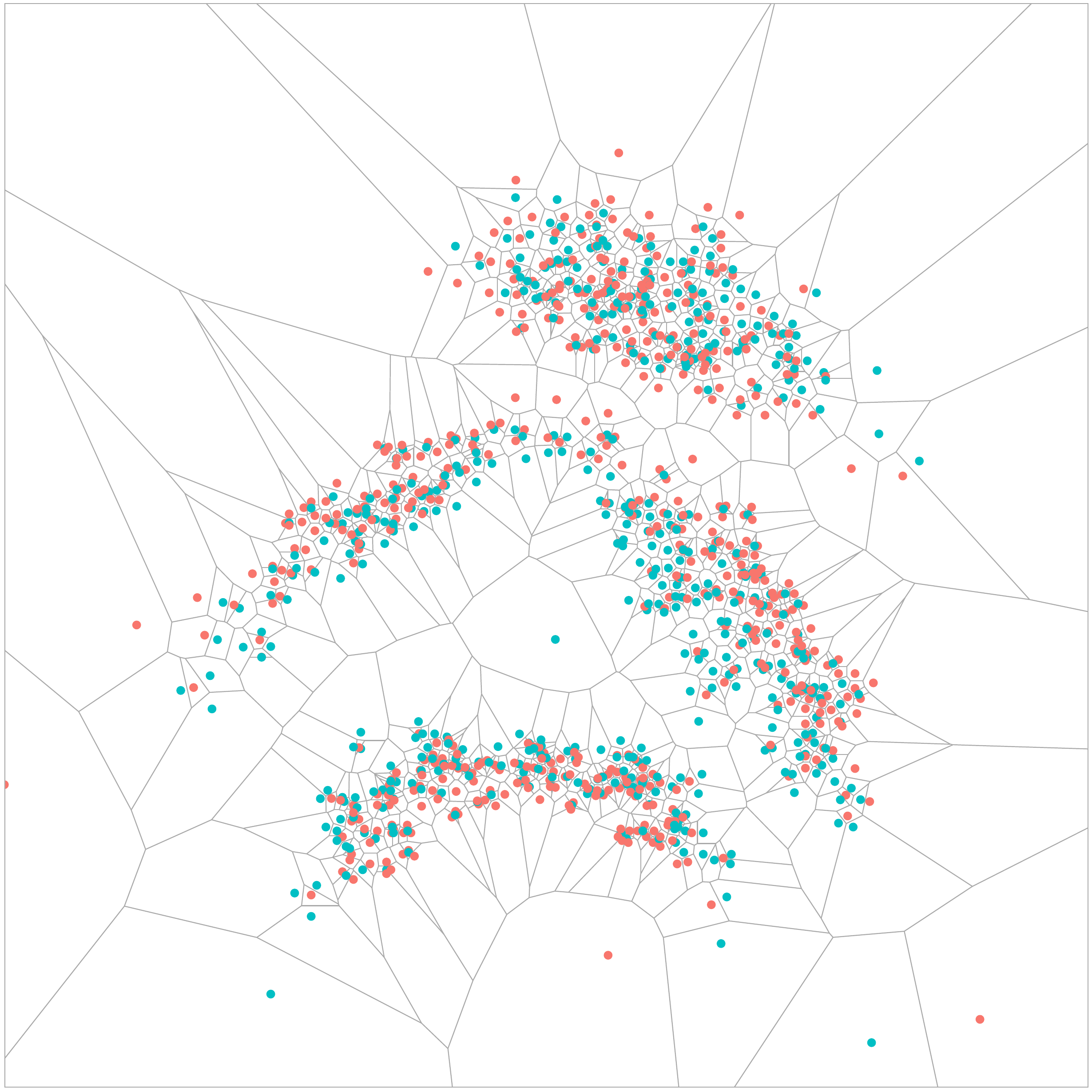}
        \caption{\FCNN (1046 pts)}\label{fig:banana:fcnn}
    \end{subfigure}%
    \begin{subfigure}[b]{.20\linewidth}
        \centering\includegraphics[width=.95\textwidth]{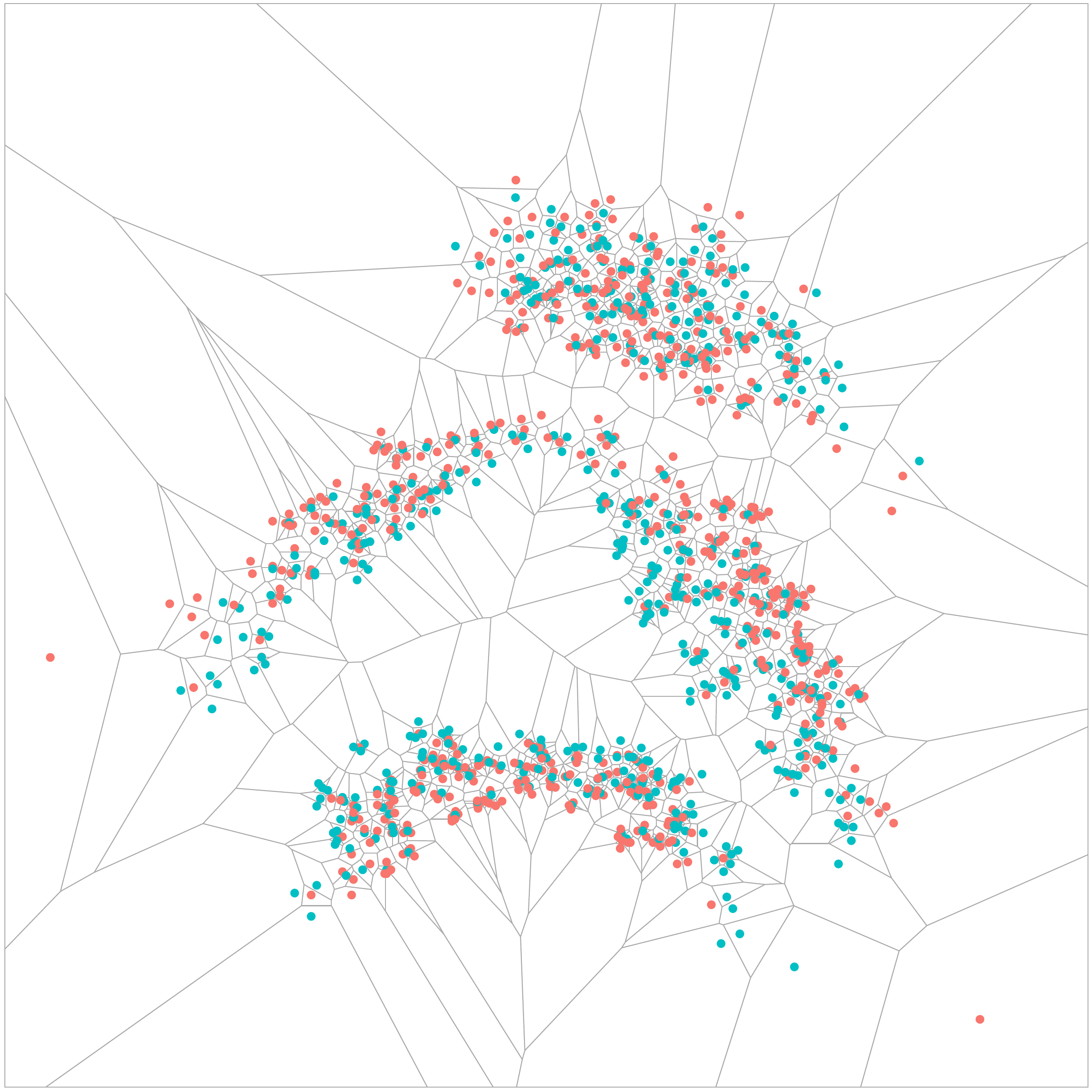}
        \caption{\MSS (1136 pts)}\label{fig:banana:mss}
    \end{subfigure}%
    \begin{subfigure}[b]{.20\linewidth}
        \centering\includegraphics[width=.95\textwidth]{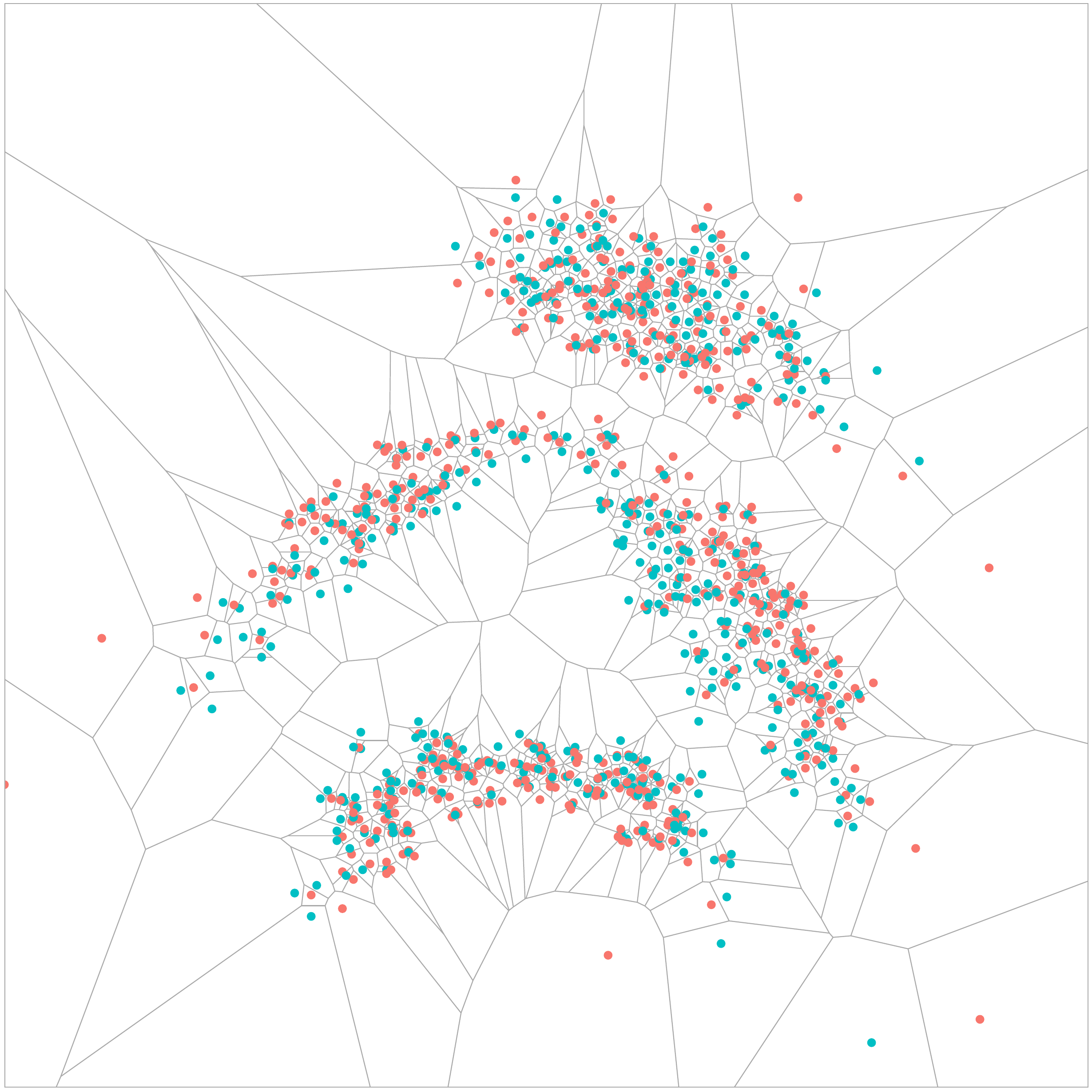}
        \caption{\RSS (1025 pts)}\label{fig:banana:rss}
    \end{subfigure}%
    \begin{subfigure}[b]{.20\linewidth}
        \centering\includegraphics[width=.95\textwidth]{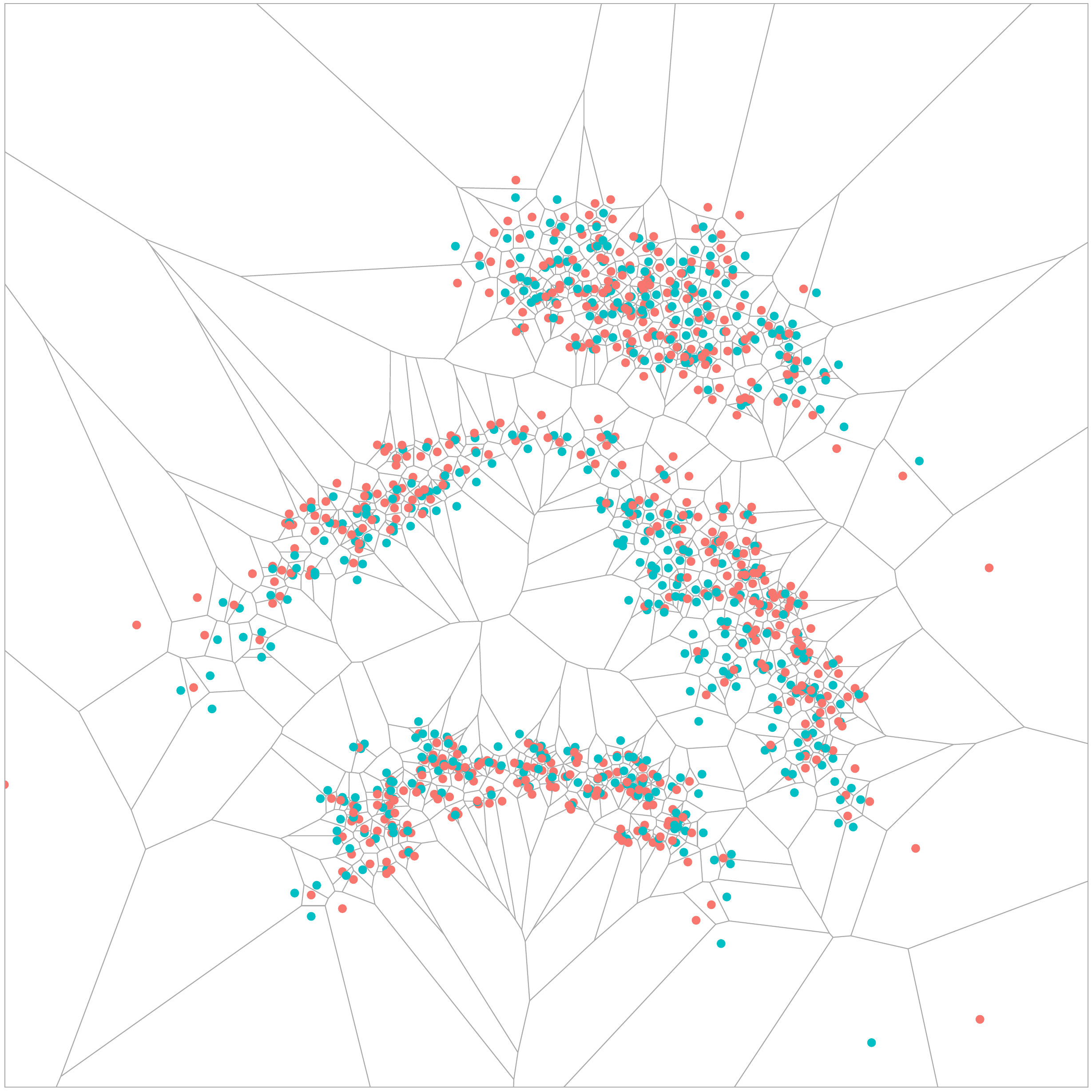}
        \caption{\VSS (1027 pts)}\label{fig:banana:vss}
    \end{subfigure}%
    \caption{Examples of the subsets selected by \FCNN, \MSS, \RSS, and \VSS, on two different training sets. Training set (a) is a set of uniformly distributed points in $\mathbb{R}^2$ of two classes: \emph{red} points lying inside a disk, and blue points lying outside. Training set (f) is a well-known benchmark from the \emph{UCI Machine Learning repository}, called \emph{Banana}, consisting of points in $\mathbb{R}^2$ of two classes, \emph{red} and \emph{blue}.}\label{fig:algexample} 
\end{figure*}

A subset $\setR \subseteq \setP$ is said to be \emph{consistent} if for all $p \in \setP$ its nearest neighbor in \setR is of the same class as $p$. Intuitively, \setR is consistent if and only if every point of \setP is correctly classified using the \NN rule over \setR. Formally, \emph{nearest-neighbor condensation} involves finding an (ideally small) consistent subset of \setP~\cite{Hart:2006:CNN:2263267.2267647}. 

Other criteria for condensation have been studied in the literature. One such criterion is known as \emph{selectivity}~\cite{ritter1975algorithm}. A subset $\setR \subseteq \setP$ is said to be \emph{selective} if and only if for all $p \in \setP$, its nearest neighbor in \setR is closer to $p$ than its nearest enemy in \setP. Clearly selectivity implies consistency, as the \NE distance in \setR of any point of \setP is at least its \NE distance in \setP. Note that neither consistency or selectivity imply that every query point of $\metricSet$ is correctly classified, just those in \setP.

The strongest criteria, known as \emph{Voronoi condensation}, consists of selecting all border points of \setP~\cite{gtoussaint84voronoi}. This guarantees the correct classification of any query point in $\metricSet$. In contrast, a consistent subset only guarantees correct classification of \setP. For the case when $\setP \subset \mathbb{R}^2$, an output-sensitive algorithm was proposed~\cite{Bremner2003} for finding all border points of \setP in $\bigOh(n \log{\numBorder})$ worst-case time. Unfortunately, it is not known how to generalize this algorithm to higher dimensions, and a straightforward approach suffers from the super-linear worst-case size of the Delaunay triangulation.

In general, it has been shown that the problems of computing consistent and selective subsets of minimum cardinality are both NP-complete~\cite{Wilfong:1991:NNP:109648.109673,Zukhba:2010:NPP:1921730.1921735,khodamoradi2018consistent}.
Thus, most research has focused on practical heuristics. For comprehensive surveys, see \cite{DBLP:conf/jcdcg/Toussaint02,Toussaint02proximitygraphs,jankowski2004comparison}.
\CNN~(Condensed Nearest-Neighbor)~\cite{Hart:2006:CNN:2263267.2267647} was the first algorithm proposed for computing consistent subsets. Even though it has been widely cited in the literature, \CNN suffers from several drawbacks: its running time is cubic in the worst-case, and the resulting subset is \emph{order-dependent}, meaning that the result is determined by the order in which points are considered by the algorithm. Alternatives include \FCNN (Fast \CNN)~\cite{angiulli2007fast} and \MSS (Modified Selective Subset)~\cite{barandela2005decision}, which produce consistent and selective subsets respectively. Both algorithms run in $\bigOh(n^2)$ worst-case time, and are order-independent. These algorithms are considered the state-of-the-art for the \NN condensation problem, subject to achieving these properties. While such heuristics have been extensively studied experimentally~\cite{jankowski2004comparison,Garcia:2012:PSN:2122272.2122582}, theoretical results are scarce. Unfortunately, to the best of our knowledge, no bounds are known for the size of the subsets generated by any of these heuristics.

More recently, an approximation algorithm called NET~\cite{gottlieb2014near} was proposed, along with almost matching hardness lower bounds for the problem. The idea is to compute a $\gamma$-net of \setP, with $\gamma$ equal to the minimum \NE distance in \setP, implying that the resulting subset is consistent. Unfortunately, while NET has provable worst-case performance, this approach allows little room for condensation, and in practice, the resulting subset can be too large. For example, on the training set in Figure~\ref{fig:circle:dataset}, \NET selects a subset of over 90\% of the points, while other algorithms select only 3\% of the points. 

%
%
%
\subsection{Our contributions}

In this paper, we propose two new \NN condensation algorithms, called \RSS and \VSS. We will establish asymptotically tight upper-bounds on the sizes of their selected subsets. Moreover, we prove that these algorithms have similar complexity to the popular state-of-the-art algorithms \FCNN and \MSS. Additionally, we also analyze the selection size of \FCNN and \MSS.
To the best of our knowledge, these represent the first theoretical results on \emph{practical} \NN condensation algorithms. The following is a summary of our contributions.

\begin{center}
\begin{tabular}{||c|c||} 
    \hline
    Algorithm & Selection size \\ [0.5ex] \hline\hline
    \RSS  & $\bigOh(\numNE\ c^{d-1})$ \\ \hline
    \VSS  & $\leq \numBorder$ \\ \hline
    \MSS~\cite{barandela2005decision}  & $\Omega(1/\eps)$ w.r.t. $\numNE$ and $\numBorder$ \\ \hline
    \FCNN~\cite{angiulli2007fast} & $\Omega(\numBorder)$ \\ \hline
\end{tabular}
\end{center}


%
%
%
\section{Results on condensation size}

One of the most significant shortcomings in research on practical \NN condensation algorithms is the lack of theoretical results on the sizes of the selected subsets. Typically, the performance of these heuristics has been established experimentally.

We establish bounds with respect to the size of two well-known and structured subsets of points: (a) the set of all \NE points of \setP of size $\numNE$, and (b) the set of border points of \setP of size $\numBorder$. Throughout the paper, we refer equally to the algorithms and their selected subsets.

%
%
%
\subsection{The state-of-the-art}

Let's begin our analysis with a state-of-the-art algorithm for the problem: \MSS or \emph{Modified Selective Subset} (see Algorithm~\ref{alg:mss}). The selection process of the algorithm can be simply described as follows: for every $p \in \setP$, \MSS selects the point with smaller \NE distance contained inside the \NE ball of $p$.

Clearly, this approach computes a selective subset of \setP, which by definition, is order-independent. \MSS can be implemented in $\bigOh(n^2)$ worst-case time. Unfortunately, the selection criteria of \MSS can be too strict, requiring one particular point to be added for each point $p \in \setP$. Note that any point inside the \NE ball of $p$ suffices for achieving selectiveness. In practice, this can lead to much larger subsets than needed.

\begin{algorithm}
 \DontPrintSemicolon
 \vspace*{0.1cm}
 \KwIn{Initial training set \setP}
 \KwOut{Condensed training set $\MSS \subseteq \setP$}
 Let $\left\lbrace p_i \right\rbrace^n_{i=1}$ be the points of \setP sorted in increasing order of NE distance $\dne{p_i}$\;
 $\MSS \gets \varnothing$\;
 $S \gets \setP$\;
 \ForEach{$p_i \in \setP$, where $i = 1\dots n$}{
  add $\gets \emph{false}$\;
  \ForEach{$p_j \in \setP$, where $j = i\dots n$}{
   \If{$p_j \in S\ \wedge\ \dist{p_j}{p_i} < \dne{p_j}$}{
    $S \gets S \setminus \left\lbrace p_j \right\rbrace$\;
    add $\gets \emph{true}$\;
   }
  }
  \If{\textup{add}}{
   $\MSS \gets \MSS \cup \left\lbrace p_i \right\rbrace$\;
  }
 }
 \KwRet{\MSS}
 \vspace*{0.1cm}
 \caption{Modified Selective Subset}
 \label{alg:mss}
\end{algorithm}

This intuition is formalized in the following theorem. Here we show that the subset selected by \MSS can select a subset of unbounded size as a function of $\numNE$ or $\numBorder$.

\begin{theorem}
Given $0<\eps<1$, there exists a training set $\setP \subset \mathbb{R}^d$ with a constant number of \NE and border points such that \MSS selects $\Omega(1/\eps)$ points.
\end{theorem}

\begin{proof}
Recall that for each point in \setP, the \MSS algorithm selects the point inside its \NE ball with minimum \NE distance. Given a parameter $0<\eps<1$, we construct a training set in $d$-dimensional Euclidean space, as illustrated in Figure~\ref{fig:mss:original}.

Create two points $r_1$ and $r_2$, and assign them to the class of \emph{red} points. Without lost of generality, the distance between these two points is 1. Let $\vec{u}$ be the unit vector from $r_1$ to $r_2$, create additional points $b_i = r_1 + \frac{i\eps}{4}\vec{u}$ for $i = \{ 1,2,\dots, 3/\eps\}$. Assign all $b_i$ points to the class of \emph{blue} points. The set of all these points constitute the training set \setP. 
It is easy to prove that \setP has only 4 \NE points and 4 border points, corresponding to $r_1, r_2, b_1$ and $b_{3/\eps}$.

\begin{figure}[h!]
    \centering
    \begin{subfigure}[b]{\linewidth}
        \centering\includegraphics[width=\textwidth]{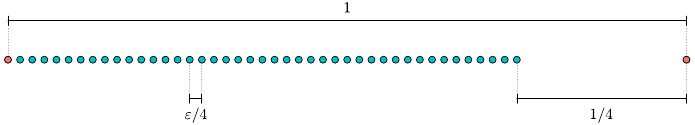}
        \caption{Initial training set of collinear points, where both the number of \NE points and the number of border points equal to 4. That is, $\numNE = \numBorder = 4$.}
		\label{fig:mss:original}
    \end{subfigure}
    
    \bigskip
    \begin{subfigure}[b]{\linewidth}
        \centering\includegraphics[width=\textwidth]{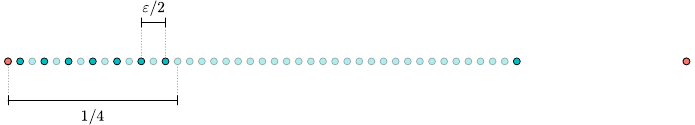}
        \caption{Subset of points computed by \MSS from the original training set (fully colored points belong to the subset, while faded points do not). The size of the subset is $\Omega(1/\eps)$.}
		\label{fig:mss:subset}
    \end{subfigure}
    \caption{Unbounded example for \MSS w.r.t. $\numNE$ and $\numBorder$.}
    \label{fig:mss}
\end{figure}

Let's discuss which points are added by \MSS for each point in \setP (see Figure~\ref{fig:mss:subset}). For points $r_1$ and $r_2$, the only points inside their \NE balls are themselves, so both $r_1$ and $r_2$ belong to the subset selected by \MSS.
For points $b_i$ with $i \leq 2/\eps$, the point with minimum \NE distance contained inside their \NE ball is $b_1$, which is also added to the subset.
Now, consider the points $b_i$ with $2/\eps < i < 5/2\eps$. Let $j = i-2/\eps$, it is easy to prove that the point with minimum \NE distance inside the \NE ball of $b_i$ is $b_{2j+1}$ (see Figure~\ref{fig:mss:subset}). Therefore, this implies that the number of points selected by \MSS equals $5/2\eps-2/\eps = 1/2\eps = \Omega(1/\eps)$.
\end{proof}

%
%
%
\subsection{A better approach}

We propose a new algorithm, called \RSS or \emph{Relaxed Selective Subset}, with the idea relaxing the selection process of \MSS, while still computing a selective subset. For a given point $p \in \setP$, while \MSS requires to add the point with smallest \NE distance inside the \NE ball of $p$, in \RSS any point inside the \NE ball $p$ suffices.

\begin{algorithm}
 \DontPrintSemicolon
 \vspace*{0.1cm}
 \KwIn{Initial training set \setP}
 \KwOut{Condensed training set $\RSS \subseteq \setP$}
 $\RSS \gets \varnothing$\;
 Let $\left\lbrace p_i \right\rbrace^n_{i=1}$ be the points of \setP sorted in increasing order of NE distance $\dne{p_i}$\;
 \ForEach{$p_i \in \setP$, where $i = 1\dots n$}{
  \If{$\dnn{p_i,\RSS} \geq \dne{p_i}$}{
   $\RSS \gets \RSS \cup \left\lbrace p_i \right\rbrace$\;
   }
 }
 \KwRet{\RSS}
 \vspace*{0.1cm}
 \caption{Relaxed Selective Subset}
 \label{alg:rss}
\end{algorithm}

The idea is rather simple (see Algorithm~\ref{alg:rss}). Points of \setP are examined in increasing order with respect to their NE distance, and we add any point whose \NE ball contains no point previously added by the algorithm. This tends to select points close to the decision boundary of \setP (see Figure~\ref{fig:circle:rss}), as points far from the boundary are examined later in the selection process, and are more likely to already contain points inside their \NE ball.

\begin{theorem}
\RSS is order-independent and computes a selective subset of \setP in $\bigOh(n^2)$ worst-case time.
\end{theorem}

\begin{proof}
By construction, every point in \setP was either added into \RSS, or has a point in \RSS inside its \NE ball. Therefore, \RSS is selective. The order-independence follows from the initial sorting step.

Let's analyze the time complexity of \RSS. The initial step requires $\bigOh(n^2)$ time for computing the NE distances of each point in \setP, plus additional $\bigOh(n\log{n})$ time for sorting the points according to such distances. The main loop iterates through each point in \setP, and searches their nearest neighbor in the current subset, incurring into additional $\bigOh(n^2)$ time. Finally, the worst-case time complexity of the algorithm is $\bigOh(n^2)$.
\end{proof}


\begin{theorem}
\label{thm:rss:size}
\RSS selects at most $\bigOh(\numNE\ (3/\pi)^{d-1})$ points.
\end{theorem}

\begin{proof}
The proof follows by a charging argument on each \NE point of \setP.
Consider a \NE point $p \in \setP$, and let $\RSS_p$ be the set of points selected by \RSS such that $p$ is their \NE. Let $p_i, p_j \in \RSS_p$ be two such points, and \wwlog say that $\dne{p_i} \leq \dne{p_j}$. By construction of the algorithm, we also know that $\dist{p_i}{p_j} \geq \dne{p_j}$. Now, consider the triangle $\triangle p p_i p_j$. Clearly, the side $p p_i$ is the larger of such triangle, and therefore, the angle $\angle p_i p p_j \geq \pi/3$.
Meaning that the angle between any two points in $\RSS_p$ with respect to $p$ is at least $\pi/3$.

By a standard packing argument, this implies that $|\RSS_p| = \bigOh((3/\pi)^{d-1})$. Finally, we obtain that $|\RSS| = \sum_p |\RSS_p| = \numNE\ \bigOh((3/\pi)^{d-1})$.
\end{proof}

For constant dimension $d$, the size of \RSS is $\bigOh(\numNE)$. Therefore, the following result implies that the upper-bound on \RSS is tight up to constant factors. Furthermore, it implies that this is the best upper-bound we can hope to achieve in terms of $\numNE$.

\begin{theorem}[Lower-bound]
There exists a training set $\setP \subset \mathbb{R}^d$ with $\numNE$ \NE points, for which any consistent subset contains $\Omega( \numNE \kern+1pt c^{d-1} )$ points, for some constant $c$.
\end{theorem}

\begin{proof}
We construct a training set \setP in $d$-dimensional Euclidean space, which contains points of two classes: \emph{red} and \emph{blue}. Consider the following arrangement of points: create a red point $p$, and take \emph{every} point at distance 1 from $p$ as a blue point. Simply, the points on the surface of a unit ball centered at $p$. 

Take any consistent subset of this training set, and consider some point $p'$ in such subset, and the bisector between $p$ and $p'$. The intersection between this bisector and the unit ball centered at $p$ describes a cap of such ball of height $1/2$. Any point located inside this cap is closer to $p'$ than $p$, and hence, correctly classified. Clearly, by definition of consistency, all points in the ball must be covered by at least one cap. 
By a simple packing argument, we know such covering needs $\Omega(c^{d-1})$ points, for some constant $c$.

So far the training set constructed has only two nearest enemy points; the red point $p$, and one blue point closest to $p$ (assuming general position). Then, we can repeat this arrangement $\numNE/2$ times, using sufficiently separated center points. This generates a training set \setP with a number of \NE points equal to $\numNE$, for which any consistent subset has size $\Omega(\numNE \kern+1pt c^{d-1})$.
\end{proof}

%
%
%

Different parameters from $\numNE$ can be used to bound the selection size of condensation algorithms. Let's consider $\numBorder$, the number of border points in the training set \setP. From the example illustrated in Figure~\ref{fig:rss:example}, we know that \RSS can select more points than $\numBorder$ (see Figure~\ref{fig:rss:example-sel}). Repeating such arrangement forces \RSS to select $\Omega(\numBorder)$ points. Yet, the question remains, at most, how many more points than $\numBorder$ can this algorithm select?

\begin{figure}[h]
    \begin{subfigure}[b]{.5\linewidth}
        \centering\includegraphics[width=.9\textwidth]{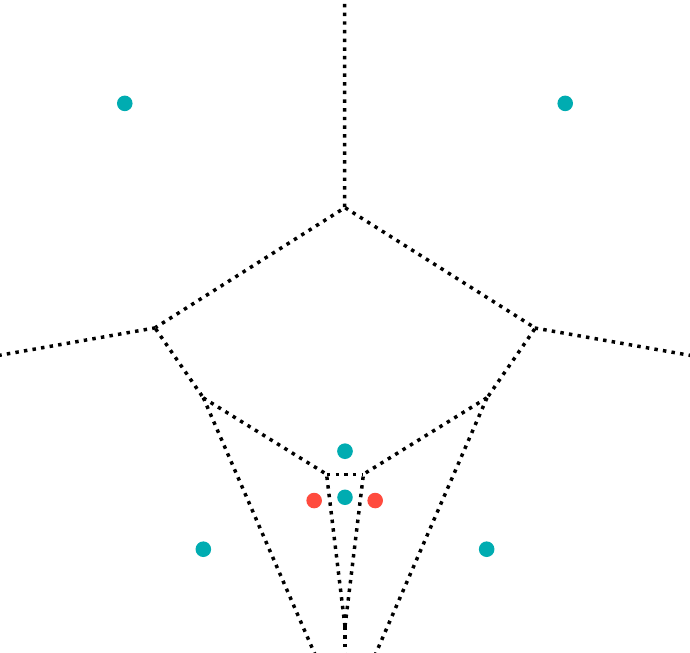}
        \caption{Point arrangement.}\label{fig:rss:example-set}
    \end{subfigure}%
    \begin{subfigure}[b]{.5\linewidth}
        \centering\includegraphics[width=.9\textwidth]{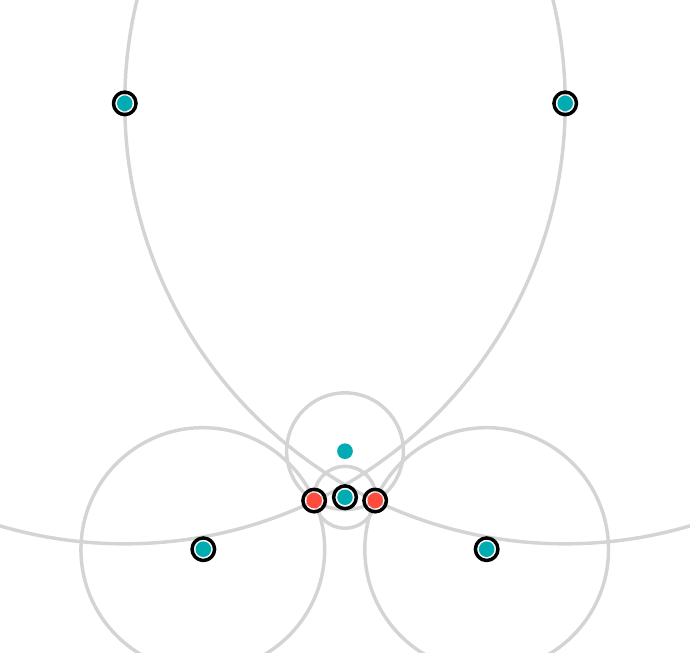}
        \caption{\RSS selection outlined.}\label{fig:rss:example-sel}
    \end{subfigure}%
    \vspace*{-1ex}
    \caption{Example where \RSS selects $\numBorder+1$ points.}\label{fig:rss:example}
\end{figure}

\begin{lemma}
\label{lemma:border}
The nearest enemy point of any point in \setP is a also a border point of \setP.
\end{lemma}

\begin{proof}
Take any point $p \in \setP$. Consider the \emph{empty} ball of maximum radius, tangent to point $\nenemy{p}$, and with center in the line segment between $p$ and $\nenemy{p}$. Being maximal, this ball is tangent to another point $p^{*} \in \setP$ (see Figure~\ref{fig:proof:border}).
Clearly, $p^{*}$ is inside the \NE ball of $p$, which implies that $p$ and $p^{*}$ belong to the same class, making $p^{*}$ and $\nenemy{p}$ enemies. By the empty ball property, this means that both $p^{*}$ and $\nenemy{p}$ are border points of \setP.
\end{proof}

\begin{figure}[h!]
    \centering
    \begin{subfigure}[b]{.45\linewidth}
        \centering\includegraphics[width=\textwidth]{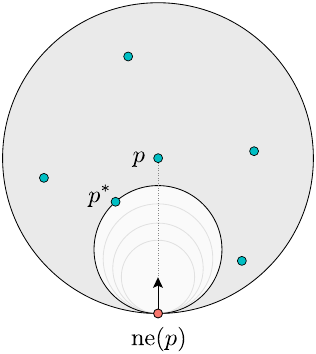}
        \caption{}
		\label{fig:proof:border}
    \end{subfigure}\hfill%
    \begin{subfigure}[b]{.45\linewidth}
        \centering\includegraphics[width=\textwidth]{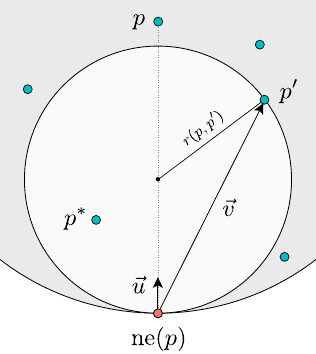}
        \caption{}
		\label{fig:vss:detail}
    \end{subfigure}
    \caption{(a) The largest empty ball tangent to $\nenemy{p}$ and center in $\overline{p\ \nenemy{p}}$, is also tangent to some point $p^{*}$, making $p^{*}$ and $\nenemy{p}$ border points. (b) Computing the radius of a ball with center in the line segment between $p$ and $\nenemy{p}$, and tangent to both $\nenemy{p}$ and $p'$.}
\end{figure}

From Lemma~\ref{lemma:border}, we know that in Euclidean space, the number of \NE points of \setP is at most the number of border points of \setP. That is, $\numNE \leq \numBorder$. While this implies an easy extension of the bound for \RSS, now in terms of $\numBorder$, it is unclear if the other factors can be improved.

Alternatively, this opens another idea for condensation. In order to prove Lemma~\ref{lemma:border}, we show that there exist at least one border point inside the \NE ball of any point $p \in \setP$. Therefore, any algorithm that only selects such border points, can guarantee to compute a selective subset of size at most $\numBorder$.
Consider then a modification of \RSS, where for each point $p_i \in \setP$, if no other point lying inside the \NE ball of $p$ has been added yet, instead of adding $p_i$ as \RSS does, we add a border point inside \NE ball of $p$. We call this new algorithm \VSS or \emph{Voronoi Selective Subset} (see Algorithm~\ref{alg:vss}).

\begin{algorithm}
 \DontPrintSemicolon
 \vspace*{0.1cm}
 \KwIn{Initial training set \setP}
 \KwOut{Condensed training set $\VSS \subseteq \setP$}
 $\VSS \gets \varnothing$\;
 Let $\left\lbrace p_i \right\rbrace^n_{i=1}$ be the points of \setP sorted in increasing order of NE distance $\dne{p_i}$\;
 \ForEach{$p_i \in \setP$, where $i = 1\dots n$}{
  \If{$\dnn{p_i,\RSS} \geq \dne{p_i}$}{
   Find a border point that lies inside the \NE ball of $p_i$ and add it to \VSS\;
   }
 }
 \KwRet{\VSS}
 \vspace*{0.1cm}
 \caption{Voronoi Selective Subset}
 \label{alg:vss}
\end{algorithm}

\begin{theorem}
\VSS computes a selective subset of \setP of size at most $\numBorder$ in $\bigOh(n^2)$ worst-case~time.
\end{theorem}

\begin{proof}
By construction, for any point in $p \in \setP \setminus \VSS$ the algorithm selected one border point inside the \NE ball of $p$. This implies that the resulting subset is selective, and contains no more than $\numBorder$ points.

Now, we describe an efficient implementation of \VSS. Recall that for every point $p \in \setP \setminus \VSS$, the algorithm finds a border point inside its \NE ball.
Without loss of generality implement \VSS to compute the point $p^{*}$ that minimizes the radius of an empty ball tangent to both $\nenemy{p}$ and $p^{*}$, and center in the line segment between $p$ and $\nenemy{p}$.
For any given point $p'$ inside the \NE ball of $p$, denote $r(p,p')$ to be the radius of the ball tangent to $p'$ and $\nenemy{p}$ and center in the line segment between $p$ and $\nenemy{p}$. As illustrated in Figure~\ref{fig:vss:detail}, let vectors $\vec{u} = \frac{p-\nenemy{p}}{\|p-\nenemy{p}\|}$ and $\vec{v} = p'-\nenemy{p}$, the radius of this ball can be derived from the formula $r(p,p') = \| \vec{v} + r(p,p')\vec{u} \|$ as $r(p,p') = \frac{\vec{v}\cdot\vec{v}}{2\vec{u}\cdot\vec{v}}$.

As $p^{*}$ is defined as the point that minimizes such radius, a simple scan over the points of \setP suffices to identify the corresponding $p^{*}$ for any point $p \in \setP$. Therefore, this implies that \VSS can be computed in $\bigOh(n^2)$ worst-case time.
\end{proof}

%
%
%
\subsection{What about FCNN?}

\begin{figure*}[t]
    \begin{minipage}{0.7\linewidth}
        \begin{subfigure}[b]{\linewidth}
            \centering\includegraphics[width=\textwidth]{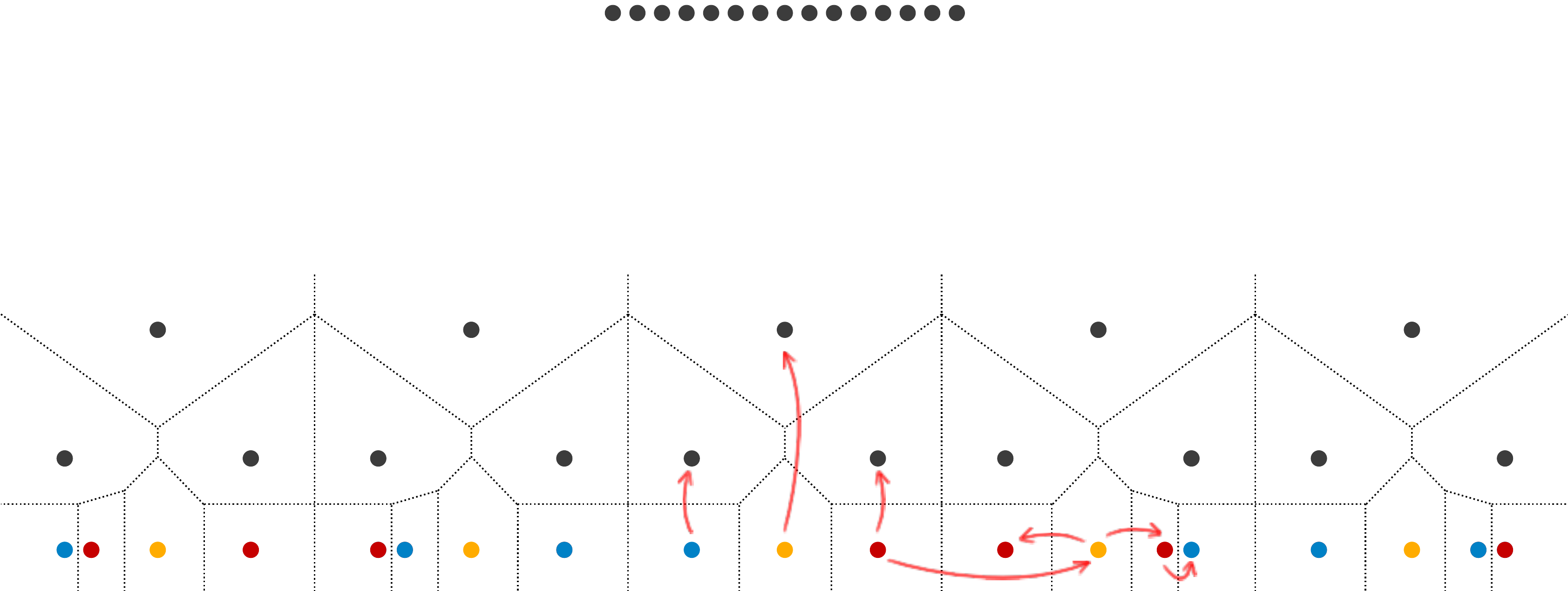}
            \caption{Entire arrangement of points.}\label{fig:fcnn:all}
        \end{subfigure}%
    \end{minipage}
    \hspace{6ex}
    \begin{minipage}{0.22\linewidth}
        \begin{subfigure}[b]{\linewidth}
            \centering\includegraphics[width=\textwidth]{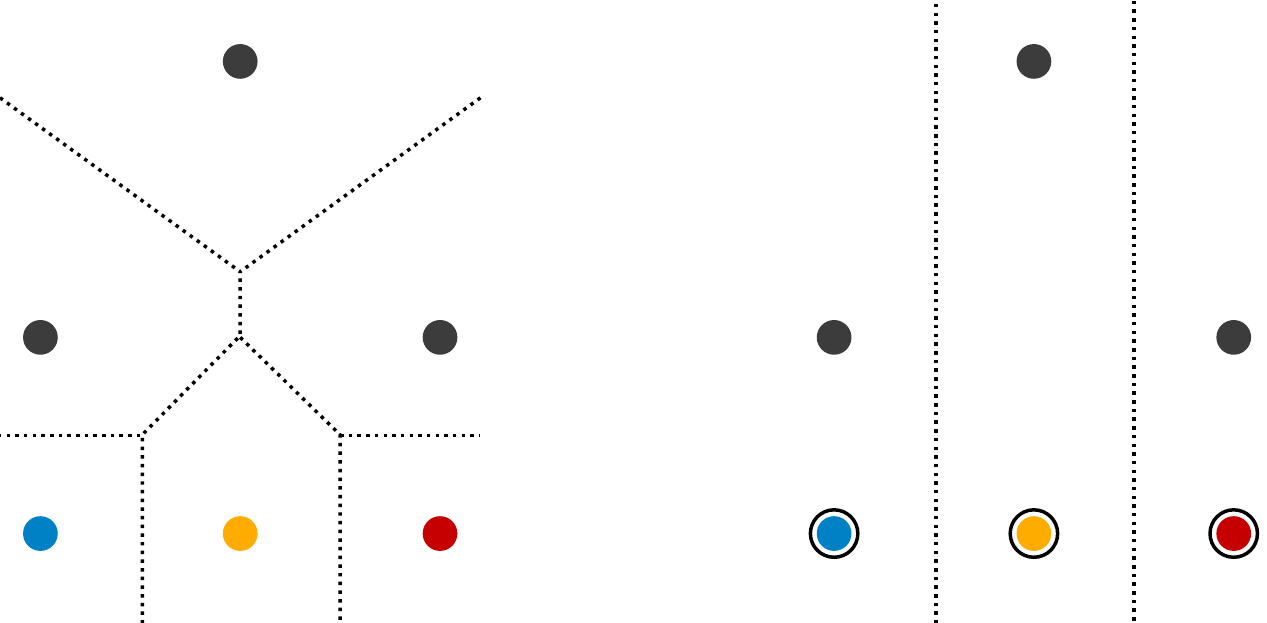}
            \caption{Middle arrangement.}\label{fig:fcnn:mid}
        \end{subfigure}\\[4ex]
        \bigskip
        \begin{subfigure}[b]{\linewidth}
            \centering\includegraphics[width=\textwidth]{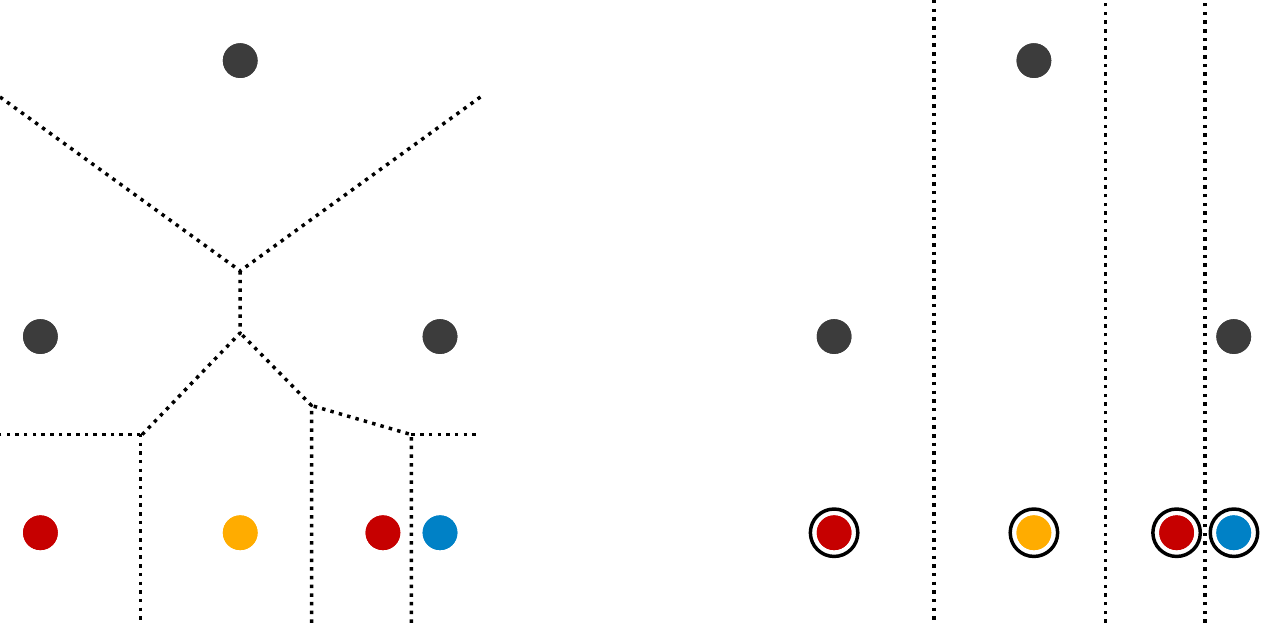}
            \caption{Side arrangement.}\label{fig:fcnn:side}
        \end{subfigure}%
    \end{minipage}
    \vspace*{-3ex}
    \caption{Example of a training set $\setP \subset \mathbb{R}^2$ for which \FCNN selects $\Omega(\numBorder)$ points.}
\end{figure*}

\FCNN or \emph{Fast Condensed Nearest-Neighbor} is yet ano\-ther popular state-of-the-art algorithm for the \NN condensation problem. In contrast with \MSS, which finds selective subsets, \FCNN selects consistent subsets of \setP.

Let's now describe the selection process of \FCNN (see Algorithm~\ref{alg:fcnn}). Essentially, \FCNN maintains a subset of \setP, which is updated in each iteration, by adding points that are incorrectly classified using the current subset. The iterations stop when all points of \setP are correctly classified by the current subset, that is, when \FCNN is consistent. Starting with the centroids of each class, set $S$ contains some misclassified points from $\setP \setminus \FCNN$ that will be added in the next iterartion. How does the algorithm decide which points to include in $S$? Define $\textup{voren}(p,\FCNN,\setP)$ as the set of enemy points of $p$ in \setP, whose \NN in \FCNN is $p$, that is, the set $\{ q \in \setP \mid l(q) \neq l(p) \wedge \nenemy{q,\FCNN} = p \}$. Then, for each point $p \in \FCNN$, the algorithm selects one representative out of its corresponding $\textup{voren}(p,\FCNN,\setP)$, which is usually defined as the \NN to $p$.

\begin{algorithm}
 \DontPrintSemicolon
 \vspace*{0.1cm}
 \KwIn{Initial training set \setP}
 \KwOut{Condensed training set $\FCNN \subseteq \setP$}
 $\FCNN \gets \varnothing$\;
 $S \gets \textup{centroids}(\setP)$\;
 \While{$S \neq \varnothing$}{
  $\FCNN \gets \FCNN \cup S$\;
  $S \gets \varnothing$\;
  \ForEach{$r \in \FCNN$}{
   $S \gets S \cup \{ \textup{rep}(p,\textup{voren}(p,\FCNN,\setP)) \}$\;
  }
 }
 \KwRet{\FCNN}
 \vspace*{0.1cm}
 \caption{Fast Condensed Nearest-Neighbor}
 \label{alg:fcnn}
\end{algorithm}

\begin{theorem}
There exists a training set $\setP \subset \mathbb{R}^d$ with $\numBorder$ border points, for which \FCNN selects $\Omega(\numBorder)$ points. 
\end{theorem}

\begin{proof}
Consider the arrangement in Figure~\ref{fig:fcnn:mid} (left), consisting of points of 4 classes. The centroids of the blue, yellow, and red classes are the only points labeled as such. By placing a sufficient number of black points far at the top of this arrangement, we avoid their centroid to be any of the three black points in the arrangement. Beginning with the centroids, the first iteration of \FCNN would have added the points outlined in Figure~\ref{fig:fcnn:mid} (right). Now each of these points have one black point inside their Voronoi cells, and therefore, these black points will be the representatives added in the second iteration. This small example, with $\numBorder=5$, shows how to force \FCNN to add all the border points plus two internal points. Out of these two internal black points, one is the centroid added in the initial step. The remaining internal black point, however, was added by the algorithm during the iterative process. This scheme can be extended to larger values of $\numBorder$, without increasing the number of classes.

The previous is the first building block of the entire training set, shown in Figure~\ref{fig:fcnn:all}. To this ``middle'' arrangement, we append ``side'' arrangements of points, as the one illustrated in Figure~\ref{fig:fcnn:side}, which will have similar behavior to the middle arrangement. This particular side arrangement will be appended to the right of the middle one, such that the distance between the red points is greater than the distance from the yellow to the red point. Every time we append a new side arrangement, its blue and red labels are swapped. The arrangements appended to the left side are simply a horizontal flip of the right arrangement. Now, the behavior of \FCNN in such a setup is illustrated with the arrows in Figure~\ref{fig:fcnn:all}. The extreme point of the previous arrangement adds the yellow point at the center of the current arrangement, which then adds the red point next to the blue point, as is closer than the other red point. Next, this red point adds the blue point, and the yellow point adds the remaining red point. Finally, the Voronoi cells of these points will look as shown in Figure~\ref{fig:fcnn:side} (right), and in the next iteration, the tree black points will be added.

After adding side arrangements as needed (same number of the left and right), it is easy to show that the centroids are still the tree points in the middle arrangement and the black point at the top (by adding a sufficient number of black points in the top cluster). This implies than \FCNN will be forced to select $\Omega(\numBorder)$ points.
\end{proof}

While this example sheds light on the selection behavior of \FCNN, an upper-bound is still missing. Based on the following lemma, we conjecture that \FCNN selects at most $\bigOh(\numNE\ c^{d-1})$~points, for some constant $c$.

\begin{lemma}
Consider a point $p$ selected by \FCNN. Then, the number of representatives of $p$ selected throughout the algorithm does not exceed $\bigOh((3/\pi)^{d-1})$ points.
\label{lemma:fcnn}
\end{lemma}

\begin{proof}
This proof follows from similar arguments to the ones described in the proof of Theorem~\ref{thm:rss:size}. Consider $p_1, p_2 \in \FCNN$ to be two points added to by the algorithm as representatives of $p$. Without loss of generality, $p_1$ was added before $p_2$, implying that $\dist{p}{p_1} \leq \dist{p}{p_2}$. By construction, if $p_2$ was added as a representative of $p$, and not of $p_1$, we also know that $\dist{p}{p_2} \leq \dist{p_1}{p_2}$. From this, consider the triangle $\triangle p p_1 p_2 $ and the angle $\angle p_1 p p_2$. This is the largest angle of the triangle, meaning that $\angle p_1 p p_2 \geq \pi/3$. Finally, by a standard packing argument, there are at most $\bigOh((3/\pi)^{d-1})$ such points.
\end{proof}

%
%
%
\section{Open problems}

A few key questions remain unsolved:
\begin{itemize}[noitemsep,topsep=1pt]
    \item In terms of $\numBorder$, our best upper-bound on the selection size of \RSS is not tight. Can it be improved?
    \item Is it possible to prove an upper-bound on the selection size of \FCNN in terms of either $\numNE$ or $\numBorder$?
\end{itemize}


\footnotesize
\bibliographystyle{abbrv}
\bibliography{nnc}

\end{document}